\documentclass[twocolumn,final,journal,comsoc,10pt]{IEEEtran}
\IEEEoverridecommandlockouts
\usepackage[caption=false]{subfig}
\usepackage{mathtools}
\usepackage[T1]{fontenc} 
\usepackage{amsmath,amsthm,amssymb,amsfonts}
\usepackage[cmintegrals]{newtxmath}
\usepackage{bm} 
\usepackage{makecell}

\newtheorem{linTheo}{Proposition}
\begin{document}

\title{GROVE: A Cost-Efficient Green Radio over Ethernet Architecture for Next Generation Radio Access Networks \\
\thanks{This work is supported by the Turkish Directorate of Strategy and Budget under the TAM Project number 2007K12-873.}
}

\author{\IEEEauthorblockN{Turgay Pamuklu\IEEEauthorrefmark{1} and
Cem Ersoy\IEEEauthorrefmark{2}}\\
\IEEEauthorblockA{NETLAB, Department of Computer Engineering,\\
Bogazici University,
Istanbul, Turkey\\
Email: \IEEEauthorrefmark{1}turgay.pamuklu@boun.edu.tr,
\IEEEauthorrefmark{2}ersoy@boun.edu.tr}}

\maketitle

\makeatletter
\def\ps@IEEEtitlepagestyle{%
  \def\@oddfoot{\mycopyrightnotice}%
  \def\@oddhead{\hbox{}\@IEEEheaderstyle\leftmark\hfil\thepage}\relax
  \def\@evenhead{\@IEEEheaderstyle\thepage\hfil\leftmark\hbox{}}\relax
  \def\@evenfoot{}%
}
\def\mycopyrightnotice{%
  \begin{minipage}{\textwidth}
  \centering \scriptsize
  \copyright~20XX IEEE.  Personal use of this material is permitted.  Permission from IEEE must be obtained for all other uses, in any current or future media, including reprinting/republishing this material for advertising or promotional purposes, creating new collective works, for resale or redistribution to servers or lists, or reuse of any copyrighted component of this work in other works. ACCEPTED ARTICLE DOI: 10.1109/TGCN.2020.3042121
  \end{minipage}
}
\makeatother

\begin{abstract}
Centralized/Cloud Radio Access Network (C-RAN) comes into prominence to reduce the rising energy consumptions and maintenance difficulties of the next-generation networks. However, C-RAN has strict delay requirements, and it needs large fronthaul bandwidth. Function splitting and Radio over Ethernet are two promising approaches to reduce these drawbacks of the C-RAN architecture. Meanwhile, the usage of renewable energy sources (RESs) in a C-RAN boosts the energy-efficiency potential of this network. In this paper, we propose a novel model, which is called Green Radio OVer Ethernet (GROVE), that merges these three approaches to maximize the benefits of C-RAN while maintaining the economic feasibility of this architecture. We briefly explain this model and formulate an operational expenditure minimization problem by considering the several restrictions due to the network design and the service provisioning. Then we linearize the quadratic routing decision constraints in the problem to solve it with a mixed-integer linear programming (MILP) solver. Results show that it is cost-effective to choose routing, function splitting, and RES decisions together. Our solution surpasses classical disjoint approaches for all studied cases. Besides, we provide a network scalability analysis to determine the MILP solver's limits for larger network topologies.
\end{abstract}

\section{Introduction}
\par A Mobile network operator (MNO) requires to serve users with several new small cells to provide larger bit rates in the next-generation mobile communication networks. The current long term evolution (LTE), in which a base station combines baseband (BBU) and radio frequency (RF) processing units together, is not a cost-effective architecture for an ultra-dense small cell network configuration \cite{Alhumaima2018}. As an alternative, Centralized/Cloud Radio Access Network (C-RAN) architecture has several benefits, such as energy-efficiency and ease of maintenance due to the centralization of the BBUs in a central unit (CU). Thus, an MNO reduces its operational expenditure (OpEx), which increases annually as a result of the increasing amount of base stations \cite{Garcia-Saavedra2018a}. The benefits of C-RAN are not limited by its scalability and multiplexing gain capability. This architecture is also a promising approach to increase the spectral efficiency by simplifying the coordinated multipoint (CoMP) technique \cite{Shehata2018}. 
\par Despite the advantages of C-RAN, increasing end-to-end delay and high-bandwidth requirement in optical fronthaul links between the CU and distributed units (DUs) make this architecture infeasible and uneconomical \cite{Dotsch2013}. Splitting the BBU functions between the CU and DUs is a promising approach to reduce these drawbacks in a C-RAN. Deciding the point where to break the function chain means how many functions are processed at the high powered CU instead  of leaving them at a local DU \cite{Larsen2019}. Besides, these split decisions may be dynamically adjusted by the MNO according to the network variations such as the daily data demand profile of the users. Therefore, this approach improves not only the overall network performance but also the quality of service (QoS) of the users in the network.  Briefly, this method is a trade-off between the energy-efficiency and reducing the delay and bandwidth requirements by choosing the weight of centralization \cite{SmallCellForum2016}.
\par Another cost of the C-RAN architecture is the capital expenditure (CapEx) which originates from the newly constructed optical fiber links between the CU and DUs \cite{Zhang2020}. The IEEE 1914 Next Generation Fronthaul Interface (xhaul) (NGFI) Working Group propose Radio over Ethernet (RoE) to reduce these costs \cite{Group}. Their standard document details this approach, in which the radio traffic between a DU and the CU is encapsulated in Ethernet frames on a multihop mesh network topology \cite{Communications2018}. This approach is more economical than the dedicated links approach by using the advantage of aggregating the traffic of different DUs in the same network lines. Moreover, this network topology may also be integrated with the backhaul network for additional cost-efficiency \cite{Gonzalez-diaz2019}. On the other hand, the RoE approach complicates the function splitting problem; thus, extra efforts and joint optimization methods are needed to provide cost-efficient solutions.
\begin{figure*}
\centering
\includegraphics[width=0.8\textwidth]{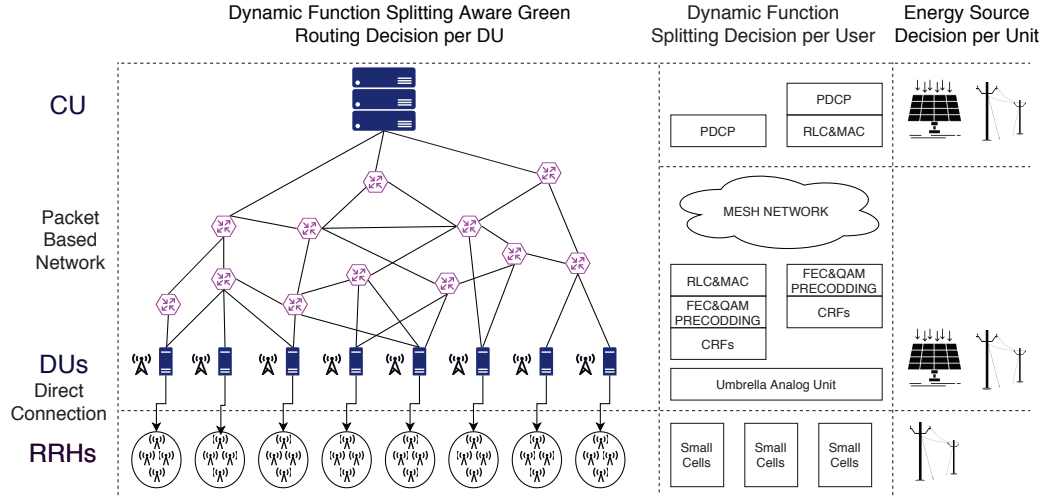}
\caption{\label{fig:arch} GROVE: Green Radio OVer Ethernet system architecture. Combining the three fundamental decisions of a next-generation wireless network (CU: central unit, DU: distributed unit, RRH: remote radio head, PDCP: packet data convergence protocol, RLC: radio link control, MAC: medium access control, FEC: forward error correction, QAM: quadrature amplitude modulation, CRF: cell-related functions).}
\end{figure*}
\par Using renewable energy sources (RESs) as an alternative to the on-grid energy in a C-RAN architecture is another promising approach to reduce the OpEx of an MNO. Besides, RESs have two critical drawbacks. First, renewable energy should be stored in a storage system such as lithium-ion batteries to use it efficiently. Nevertheless, these systems have limited storing capacity and increasing their capacity impact directly the CapEx. Thus, we need to opt for the solutions which promote the practical usage of the valuable storage systems \cite{Pamuklu2018}. Second, RESs are sporadic and unstable sources, and they provide highly unpredicted energy across time and space \cite{Chiang2018}. Therefore, these sources usually need to be supported by a reliable source like on-grid electricity.    
\par In this paper, we aim to reduce the OpEx of an MNO by combining these three new concepts in a C-RAN: splitting the BBU functions between a CU and several DUs, transferring the traffic between them with the RoE approach, and using RESs as an alternative energy source in these units. Figure~\ref{fig:arch} presents this architecture. To the best of our knowledge, there are no studies to model this problem and find a joint solution to reduce the OpEx. Our contributions can be summarized as follows:
\begin{enumerate}
\item We model the Green Radio OVer Ethernet (GROVE) concept. This novel concept is a promising approach to make a C-RAN architecture cost-effective for an MNO.
\item We formulate an optimization problem that aims to reduce the OpEx of this new model. We present that in this model, we have to jointly make decisions for the function splitting, dynamic routing decisions between the CU and DUs, and considering to use the RESs efficiently.
\item We linearize the quadratic constraints that proceed from choosing the function splitting and routing decisions concurrently. Therefore, the problem is transformed into a new form that can be solved with a MILP Solver.
\item We experiment with different traffic loads to show the performance of the solution for diverse city populations. Besides, we use real solar data and examine our solution for different seasons to see the impact of seasonal changes in four different geographical areas in the world, which have significantly different solar radiation distributions.
\item  We examine the limits of the MILP solver for larger network problems.
\end{enumerate}
\par The remainder of this paper is organized as follows. The related work is discussed in the second section. Then, we define the GROVE system model and its cost optimization problem in the third and fourth sections, respectively. In the fifth section, we present the results of the computational experiments, followed by the concluding remarks and future works in the last section.
\section{Related Work}
\subsection{Function Splitting Approaches for C-RAN}
\par One of the key performance metrics of a C-RAN architecture is the multiplexing gain that comes from centralizing the BBU functions in the CU. Thus, we can get rid of the unnecessary energy consumption of underutilized BBUs originate from low traffic loads. Checko et al. advance the multiplexing gain analysis by investigating it for function splitting approaches \cite{Checko2016}. They form their proposed architecture in two ways: an N-dimensional Markovian process model and a discrete event simulation model. Then, for each model, they provide the impact of different splitting decisions on the quantitative multiplexing gain results, in terms of energy and cost minimization. Wang et al. introduce a new approach to the function splitting concept, in which some of the functions are processed in an edge cloud (EC) as an alternative to the CU \cite{Wang2017}. In this architecture, the ECs may serve more than one RRHs to increase the QoS. They call this architecture "Hybrid C-RAN," and their objective is jointly minimizing the total energy consumption in the network and the bandwidth of the midhaul which provides the connections between the CU and ECs. In their next paper, they enhance their study by analyzing the effect of traffic load on their findings \cite{Alabbasi2018}.
\par Some studies prefer to model the function splitting problem as a graph, in which the BBU functions are represented as nodes, and the connections between these nodes are shown as weighted edges. Mharsi et al. choose the weight of the edges as the latency requirement \cite{Mharsi2018}. The start nodes of their graph are the antennas, and the end node is the CU. Thus, they perform splitting decisions for all data flows between the antennas and the core node in their studied network. They have a multi-objective function that jointly minimizes the sum of the end-to-end latencies and the total number of CPUs used in their proposed network. Meanwhile, Liu et al. choose more than one metric for the weight of the edges, which are the computational costs and the fronthaul link costs \cite{Liu2015}. Then they characterize the tradeoff between these two parameters while choosing the delay as a constraint in their problem formulation. 
\par Shehata et al. focus on static function splitting options. They investigate the effect of these options on reducing the energy consumption and the needed giga operations per second (GOPS) in the network \cite{Shehata2018}. They describe the difference between the local BBU architecture and the BBU pool in C-RAN. Also, they explain the layers that are processed in these BBUs and the impact of several split options between these layers. Their analytical model starts with User-Evolved-Node-B (eNB) assignment according to the highest received signal strength indicator at the user side. Next, they schedule the physical resource blocks (PRBs) among the users. Then, they provide a detailed energy consumption model of BBUs in a classical distributed-RAN and a C-RAN. Their experiments demonstrate the improvement of the system performance in each split option for different geographical type areas. Harutyunyan et al. suggest a virtual network embedding (VNE) approach in their papers for the purpose of finding the optimum place for splitting the BBU functions \cite{Harutyunyan2018}. In their first study, they formulate a problem in which they jointly minimize the interference and the fronthaul bandwidth. In their second study, they combine the problem of choosing the optimum places for the BBU functions with minimizing the number of used millimeter-wave wireless fronthaul links in this VNE model \cite{Harutyunyan2019}.
\subsection{Integrating Radio over Ethernet with C-RAN methods}
\par The Radio over Ethernet (RoE) approach reduces not only the operating costs but also the planning costs of an MNO. Thus, several recent studies implement this approach to the function splitting problem to improve the feasibility of their proposed models. Garcia-Saavedra et al. present a decision-making engine, called Wizhaul, that jointly choose the flow paths and the weight of the function centralization in a CU \cite{Garcia-Saavedra2018}. By using this engine, they provide solutions for both network planning and operating phases. In their recent paper, they also include multi-access edge computing in their problem models. These studies are promising guides to integrate the RoE with the function splitting approach to improve the edge computing performance in a C-RAN \cite{Garcia-Saavedra2018a,Garcia-Saavedra2018b}. Meanwhile, they provide a detailed analysis of the crosshaul approach in a separate study. This approach combines the fronthaul and backhaul networks as a joint packet-based network for further reducing the network costs \cite{Gonzalez-diaz2019}.
\par Chang et al. focus on three inter-related problems in an RoE network \cite{Chang2016}. The first problem aims to packetize the BBU processed data in Ethernet frames according to choosing the function splitting decisions. The second one deals with the scheduling of these frames, and the last one ensures a hybrid automatic repeat request based timing constraint while providing a solution for the first and second problems. In their following work, they provide a detailed performance analysis of several key indicators such as the network throughput and user satisfaction \cite{Chang2017}. Ojaghi et al. highlight the importance of network slicing to improve the throughput \cite{Ojaghi2019}. They combine it with RoE connections and target to optimize the computational cost and the throughput of the whole network. Diez et al. aim to minimize the total end-to-end latency in their packet-based network by providing a connection for each RRH \cite{Diez2019}. They compare their solution with the fixed split and fixed scheduling approaches.
\subsection{Using Renewable Energy Sources in a C-RAN architecture}
\par Alameer et al. provide a solution for a classical C-RAN architecture \cite{Alameer2016}. They represent a RES system with a queuing model, in which RESs are deployed in each RRHs and BBUs.Their objective is to minimize the overall energy consumption of this model by considering the QoS. Guo et al. focus on a similar problem, in which they represent the system as an MINLP problem \cite{Guo2018}. They propose a two-phase heuristic to reduce the brown energy consumption. 
\par Although the function splitting is a new concept, Temesgene et al. integrate the RESs in this concept and provide a detailed energy consumption analysis in this system \cite{Temesgene2019}. They propose a solution for an offline problem to reduce the on-grid energy consumption. Then, in their following study, they provide a solution for an online problem that dynamically change the splitting decisions according to the traffic load and the harvested energy \cite{Temesgene2018}. Meanwhile, Wang et al. propose a novel model to maximize the throughput of the network by solving the function splitting problem with RESs \cite{Wang2018}. On the contrary, Ko et al. choose the throughput as a constraint in their problem \cite{Ko2018}. Then they target to reduce the overall on-grid energy consumption in the network by using RESs.
\par In our previous study, we integrated a RES model with a hybrid C-RAN architecture and formulated an OpEx minimization problem by considering the network constraints \cite{Pamuklu2020}. After describing the system model, we provided cost-effective solutions for MNOs. Unlike prior work, in this paper, we integrate the RoE concept into this system. Instead of connecting the DUs to CU with direct fiber links, we implement a packet-based network between these units. Then, we solve an optimal routing problem. The results show that efficient routing decisions are mandatory to reduce the cost of MNOs. In the next section, we describe the system model with the RoE network in detail.
\begin{figure*}
\centering
\subfloat[\label{fig:ti1} Time Interval 1.]{
\includegraphics[width=.40\textwidth]{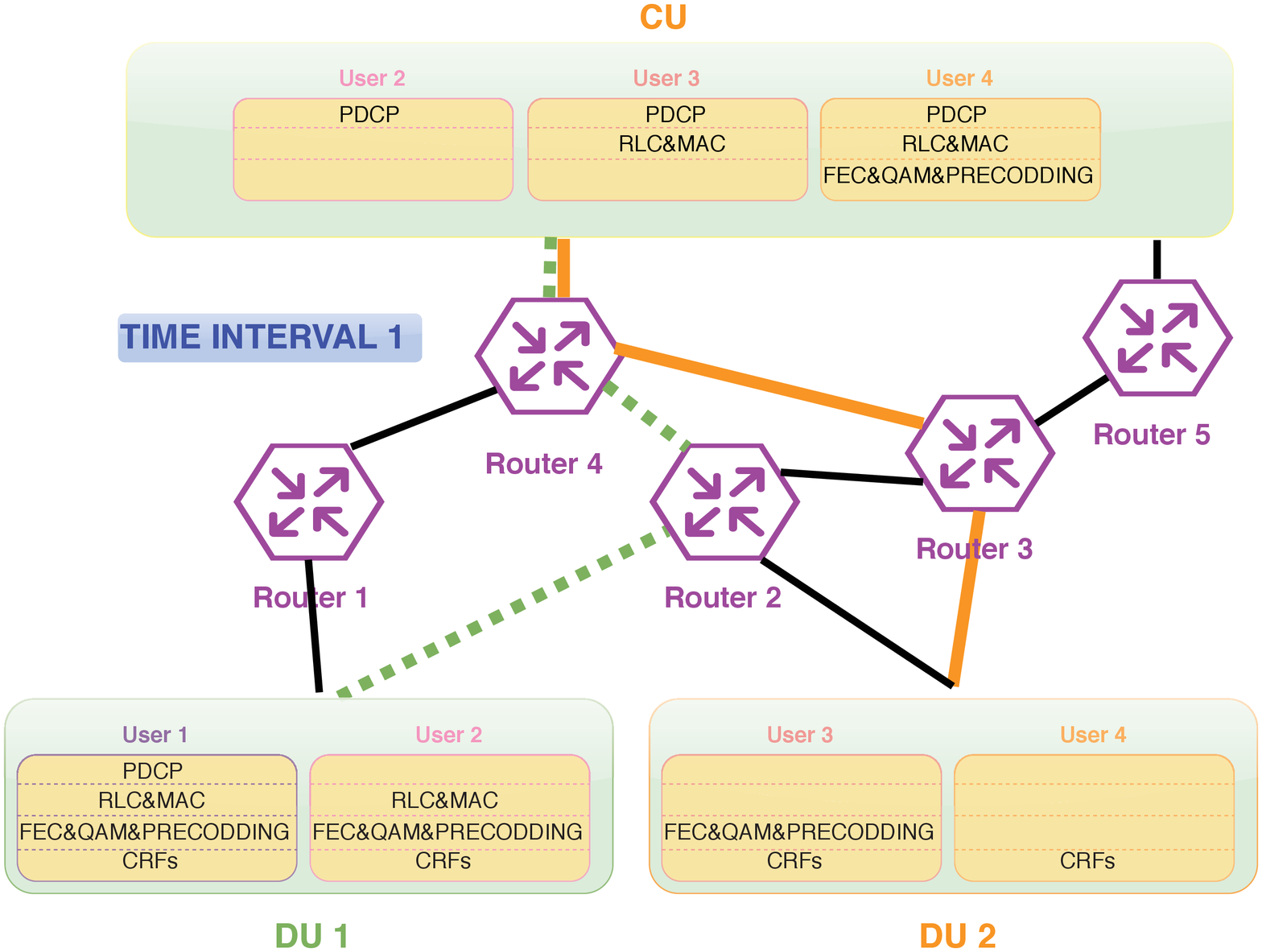}}
\subfloat[\label{fig:ti2} Time Interval 2.]{
\includegraphics[width=.40\textwidth]{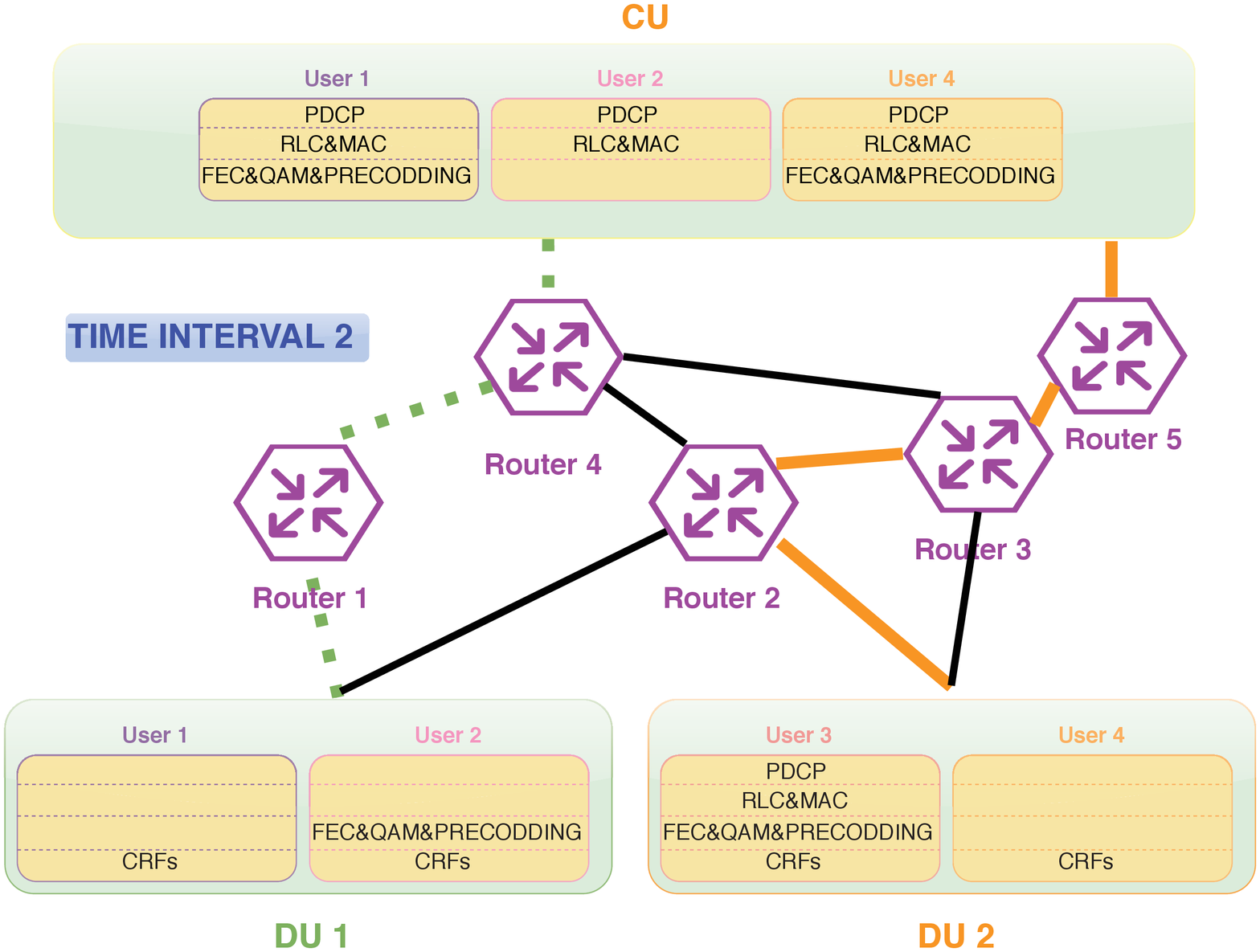}}
\caption{Dynamic functional splitting on different routing decisions. Green dashed and orange solid paths show the decisions for DU1 and DU2, respectively.}
\label{fig:splitting}
\end{figure*}
\section{GROVE System Model}
\par Figure~\ref{fig:arch} illustrates the overall architecture of the GROVE model. We consider a scenario that a set of DUs, $\mathcal{R}$, are connected to the CU through a packet-based network $G=(\mathcal{V},\mathcal{E})$, where $\mathcal{V}$ is the superset of routers, the set of DUs, and the CU; and $\mathcal{E}$ is the set of links between these components. Each link $e \in \mathcal{E}$ has a limited capacity $\Omega_{e} \geq 0$ for the downlink transmissions from CU to DUs \footnote{Although we just consider the downlink transmissions for simplicity, the network model can be straightforwardly extended to include uplink transmissions.}. This architecture reduces the CapEx of MNOs by replacing expensive point-to-point links with a mesh network. The growing number of DUs in the next-generation RANs increases the cost dramatically and leads researchers to investigate RoE approaches \cite{Communications2018}. In order to reduce the OpEx, dynamic routing decisions should be provided synchronously with the changing traffic loads which are directly related to the function splitting decisions.
\par Each DU has a Radio Frequency (RF) equipment with an antenna. Besides, there is a set of RRHs, $\mathcal{C}$, that are connected to their corresponding DU with a point-to-point millimeter-wave or dedicated fiber links \cite{Alabbasi2019}. While these RRHs do not have any BBUs, this architecture provides throughput enhancement with lower costs. Since the RRHs are geographically close to their corresponding DU, the CapEx of the MNO to connect them remains at lower values.
\par BBU functions are provided by the digital processing equipments (DPEs) in the DUs and the CU, collectively. The chain of these functions is broken at a certain split point. This point is dynamically decided for a set of time intervals in a day period for each user. These users have different data traffic loada on the DPEs, $\rho_{it}$, and the maximum delay threshold, $\mu_{it}$, according to their required service type.
\par Cell-related functions (CRFs) and user-related functions (URFs) are the two main divisions of the BBU functions \cite{Alabbasi2018}. According to the Small Cell Forum, breaking the chain after CRFs reduces the required bandwidth significantly \cite{SmallCellForum2016}. While the packet-based network between the CU and DUs should have a limited capacity to achieve a reasonable CapEx of the MNO, we prefer to process all CRFs on the DUs side. Moreover, the processing load and the bandwidth consumption of URFs vary with the user traffic demands; thus, splitting after CRFs also has multiplexing benefits due to the mesh network.  
\par We have four splitting options on the CU side that should be given by considering the routing decisions in each time interval. Figure~\ref{fig:splitting} illustrates these decisions for two different time intervals. These options are:
\begin{enumerate}
\item CU does not process any URF. This option is suitable for ultra-reliable low latency requirement traffic (URLLC), which has stringent delay requirements \cite{Larsen2019}. On the other hand, choosing this option for all traffic types literally yields a distributed RAN. For example, the enhanced mobile broadband (eMBB) energy consumption dominates the overall energy consumption in a RAN. This traffic type has high throughput demand and produces a high traffic load on DPEs \cite{Elsayed2019a}. Thus, their functions should be favored to be processed in the CU which has more energy-efficient DPEs. 
\item CU processes only packet data convergence protocol (PDCP). The functions in this protocol are not very strict about delay requirements. Besides, centralizing the PDCP functions promotes mobility across the DUs \cite{SmallCellForum2016}.
\item CU processes radio link control (RLC) and medium access control (MAC) in addition to the PDCP. Thus, it processes all URFs in the MAC layer and above. This option is more strict in terms of delay, but have several centralization advantages such as allowing traffic aggregation \cite{3GPP2017}.
\item CU processes forward error correction (FEC), quadrature amplitude modulation (QAM), and precoding processes, which means it processes all URFs including the physical layer functions. Although this split is the most energy-efficient option in this paper, in order to choose it, the delay requirement should be non-stringent. Moreover, this option needs the highest bandwidth on the mesh network, and it is not feasible to provide this option for all users and DUs concurrently.
\end{enumerate}
\par In addition to these splitting points, there are also other splitting point options that may be chosen according to MNO's requirements \cite{SmallCellForum2016}. The splitting options that provide processing of a higher amount of URFs in CU reduce the number of total active DPEs in the network due to the aggregation of these functions in the centralized DPEs. Thus,
more energy-saving will be provided by MNOs with these splitting options. However, we increase the required bandwidth in the packet-switched network and end-to-end latency with the extension in centralization \cite{Alabbasi2018}. Besides, the DUs and CU have batteries that have limited capacities. Thus, we have to dynamically change the splitting options according to the renewable energy amount in these batteries. Otherwise, MNOs lose the valuable unstored renewable energy in the unit where they do not prefer to process the functions. Thus, the network's overall energy consumption increases with inefficient function splitting decisions. Moreover, the choice of splitting options is also related to the network topology and routing decisions. For example, if the other DUs congest a network link (Figure~\ref{fig:ti1}), we have to select the splitting options that process more URFs on the DU side. Therefore, it is crucial to choose the function splitting and the routing decisions together to balance the load on each side (Figure~\ref{fig:ti2}). However, in a scenario with a higher number of DUs, it is unavoidable that many DUs should share the same links.
\begin{table}
\centering
\caption{\label{tab:Notations} Summary of the notations. }
\begin{tabular}{|c|p{6cm}| }
\hline
\textbf{Sets} & \textbf{Explanation} \\ \hline
$t\in\mathcal{T}$ & set of time intervals \\
$i\in\mathcal{I}$ & set of users \\
$c\in\mathcal{C}$ & set of RRHs \\
$d\in\mathcal{D}^{y}$ & set of DPEs ($y$ is $CU$ or $DU$)\\
$r\in\mathcal{R}$ & set of DUs \\
$f\in\mathcal{F}$ & set of URFs \\
$v\in\mathcal{V}$ & set of nodes \\
$e\in\mathcal{E}$ & set of edges, $e = (x\in\mathcal{V},y\in\mathcal{V})$ \\ \hline
\textbf{Variables}& \textbf{Explanation} \\ \hline
$m_{idft}$ & URF $f$ of UE i is hosted in DPE $d$ \\
$a_{dt}$ & DPE $d$ is active in time interval $t$ \smallskip \\ 
\makecell[cc]{$l_{rte}$\medskip} & \makecell[cl]{path decision variable indicating if \\edge $e=(x\in\mathcal{V},y\in\mathcal{V})$ is on the path \\towards DU $r$ in time interval $t$}\smallskip\\
$s^{y}_{rt}$ & green energy consumption\\
$b^{y}_{rt}$ & remaining energy in the battery\\
$p^{y}_{rt}$ & sold energy\\ \hline
\textbf{Input} & \textbf{Explanation} \\ \hline
$\rho_{it}$ & traffic load ratio of user $i$\\
$\mu_{it}$ & delay threshold of user $i$\\
$L^{y}$ & DPE function cap. ($y$ is $CU$ or $DU$)\\
$B^{y}_{r}$ & battery maximum storage capacity\\
$S^{y}_{r}$ & solar panel size\\
$G^{y}_{rt}$ & generated green energy\\
$\Psi^{y}_{rt}$ & total energy consumption\\
$\mathbb{E}_{t}$ & energy price in time interval $t$\\
$\mathbb{P}$ & sold energy penalty ratio\\
$\Omega_{e}$ & bandwidth capacity of edge (link) $e$\\
\hline
\end{tabular}
\end{table}
\par If we now move to the energy model, Figure~\ref{fig:arch} illustrates that the CU and the DUs have two different energy sources to facilitate their operations. While the first one, the solar panel, reduces the OpEx of the MNO with renewable energy, the second one, on-grid energy, becomes a reliable source in the case of the lack of insufficient green energy. The other components in the system, RRHs, have only an on-grid energy source. Since most of the 5G RRHs are expected to be indoor, the MNO can easily position them in an indoor area that is not directly exposed to solar radiation. However, if needed, RESs can also be used for RRHs.
\par The energy consumptions in the DUs ($\Psi^{DU}_{rt}$) and the CU ($\Psi^{CU}_{t}$) are formulated in Equation~\ref{eq:energyConsInDU} and Equation~\ref{eq:energyConsInCU}, respectively. $E^{DU}_{DPE}$ is the energy consumption of a DPE in a DU, and $a_{dt}$ is the binary decision that identifies the activity of DPE $d$ in time interval $t$. Besides, a DU has a static energy consumption ($E^{DU}_{STA}$) due to the RF equipment, cooling system, and the other idle energy consumptions which do not change by the DPE activity. Similarly, a CU also has a static energy consumption ($E^{CU}_{STA}$), which is similar to the DU side, except it does not include any RF equipment energy consumption.
\begin{align}
\Psi^{DU}_{rt} &= E^{DU}_{STA} + \sum\limits_{d\in\mathcal{D}_{r}}a_{dt}E^{DU}_{DPE}
\label{eq:energyConsInDU}
\\
\Psi^{CU}_{t} &= E^{CU}_{STA} + \sum\limits_{d\in\mathcal{D}_{CU}}a_{dt}E^{CU}_{DPE}
\label{eq:energyConsInCU}
\end{align}
\par The mesh network between the CU and DUs also consumes energy. However, most of this energy consumption does not change with the traffic load and remains as a static value. The only way to reduce this energy consumption is by completely switching off the routers and the links \cite{Liu2018}. Despite this, in our system model, the links are always active, and we do not need to optimize this energy consumption; thus, we do not include it in the objective function.    
\section{Minimizing the OpEx of GROVE Model}
\subsection{Problem Formulation}
\par Table~\ref{tab:Notations} outlines the notations in this section. Our primary goal is to reduce the OpEx of the GROVE system model explained in the previous section. This minimization problem can be defined as: 

\begin{align}
&\textbf{Minimize:} \notag
\\
&\sum\limits_{t\in\mathcal{T}}  \biggl [  \Psi^{CU}_{t}-s^{CU}_{t}-\mathbb{P}*p^{CU}_{t}  \notag
  \\&+  \sum\limits_{r\in\mathcal{R}}(\Psi^{DU}_{rt}-s^{DU}_{rt}-\mathbb{P}*p^{DU}_{rt}) \biggr] * \mathbb{E}_{t}
\label{eq:obj1}
\\
&\textbf{Subject to:}\footnote{Constraints should be satisfied for all time intervals ($\forall t \in \mathcal{T}$).} \notag
\\
&\sum\limits_{f\in\mathcal{F}}\sum\limits_{i\in\mathcal{I}} \rho_{it} m_{idft} < L^{CU}, \quad\forall d \in \mathcal{D}^{CU}
\label{eq:DPECapInCU}
\\
&\sum\limits_{f\in\mathcal{F}}\sum\limits_{c\in\mathcal{C}_{r}} \sum\limits_{i\in\mathcal{I}_{c}} \rho_{it}  m_{idft} < L^{DU}, \quad\forall d \in \mathcal{D}^{DU}_{r} ,\forall r\in\mathcal{R} 
\label{eq:DPECapInDU}
\\
&M * a_{dt} -\sum\limits_{f\in\mathcal{F}}\sum\limits_{i\in\mathcal{I}}m_{idft}\geq 0, \quad\forall d \in \mathcal{D}^{CU}
\label{eq:DPEActiveInCU}
\\
&M * a_{dt} - \sum\limits_{f\in\mathcal{F}}\sum\limits_{c\in\mathcal{C}_{r}} \sum\limits_{i\in\mathcal{I}_{c}}m_{idft}\geq 0, \;\;\forall d \in \mathcal{D}^{DU}_{r}, \forall r \in \mathcal{R}
\label{eq:DPEActiveInDU}
\\
&\sum\limits_{f\in\mathcal{F}}\sum\limits_{d\in\mathcal{D}^{CU}\cup\mathcal{D}^{DU}_{r}}m_{idft} = |\mathcal{F}|, \;\;\forall i\in\mathcal{I}_{c}, c\in\mathcal{C}_{r}, \forall r \in \mathcal{R}
\label{eq:FunctionAssign}
\\
&\sum\limits_{f\in\mathcal{F}} \sum\limits_{d\in\mathcal{D}^{CU}} m_{idft} < \mu_{it},\quad \forall i \in \mathcal{I}
\label{eq:DelayConst}
\\
&b^{CU}_{t} = b^{CU}_{(t-1)} - s^{CU}_{t} - p_{t} + S^{CU}G^{CU}_{t}
\label{eq:battEnergyInCU}
\\
&b^{DU}_{rt} = b^{DU}_{r(t-1)} - s^{DU}_{rt} - p^{DU}_{rt} + S^{DU}_{r}G^{DU}_{rt},\quad \forall r\in\mathcal{R}
\label{eq:battEnergyInDU}
\\
 &b^{CU}_{t} \leq B^{CU}
\label{eq:battLimitInCU}
\\
&b^{DU}_{rt} \leq B^{DU}_{r},\quad\forall r\in\mathcal{R}
\label{eq:battLimitInDU}
\\
&s^{CU}_{t} \leq  \Psi^{CU}_{t}
\label{eq:renEnMaxLimitInCU}
\\
&s^{DU}_{rt} \leq  \Psi^{DU}_{rt},\quad\forall r\in\mathcal{R}
\label{eq:renEnMaxLimitInDU}
\end{align}
\begin{multline}
\sum\limits_{y\in\mathcal{V}} l_{rt(x,y)} - \sum\limits_{y\in\mathcal{V}} l_{rt(y,x)}  =
\begin{cases} 
1, & \text{if } x\in \mathcal{V}^{DU} \\
-1,& \text{if } x\in \mathcal{V}^{CU} \\
0, & \text{otherwise,}
\end{cases}\\ \forall r \in \mathcal{R} , \forall x\in \mathcal{V}
\label{eq:routing}
\end{multline}
\begin{multline}
\label{eq:bwConstraint}
\sum\limits_{r\in\mathcal{R}} l_{rt(x,y)} \sum\limits_{c\in\mathcal{C}_{r}} \sum\limits_{i\in\mathcal{I}_{c}} \sum\limits_{d\in\mathcal{D}^{CU}} \sum\limits_{f\in\mathcal{F}} \rho_{it} m_{idft} \leq \Omega_{(x,y)},\\\forall (x,y)\in \mathcal{E}
\end{multline}
\par The OpEx of the system is the overall on-grid electricity bills of the CU and the DUs\footnote{The electricity cost of the RRHs and the maintenance cost of the GROVE system are not included in the OpEx calculation. The reason is that these costs do not change by any decision variable given in Table~\ref{tab:Notations}.}.
We have to deal with three fundamental problems to reduce these bills. First, we have to reduce the total energy consumptions in the CU ($\Psi^{CU}_{t}$) and the DUs ($\Psi^{DU}_{rt}$) by switching off as many as DPEs in these units. Second, we have to favor the usage of renewable energy ($s^{CU}_{t}$ and $s^{DU}_{rt}$) instead of on-grid energy consumption. This action principally depends on the size of solar panels and the batteries in these units and planning the use of renewable energy in these batteries in an efficient way. Otherwise, the MNO should sell this valuable energy to the grid network at a reduced price ($\mathbb{P}*p^{CU}_{t}$ and $\mathbb{P}*p^{DU}_{rt}$). Lastly, we have to consider the variation of the electricity prices in a day period ($\mathbb{E}_{t}$) and try to use renewable energy as much as possible in the time intervals that the price of the electricity is higher than the average price.
\par The first two constraints (Inequalities~\ref{eq:DPECapInCU} and~\ref{eq:DPECapInDU}) model the DPEs capacity limitations. $L^{CU}$ and $L^{DU}$ are the maximum numbers of URFs that can be executed in a DPE of a CU or a DU, respectively. As stated by Mharsi et al., a URF's processing demand correlates with the traffic load of the corresponding user; thus, the number of executed URFs in a DPE ($m_{idft}$) is multiplied with the traffic load in these constraints \cite{Mharsi2018}. Meanwhile, deciding to execute a URFs in a DPE leads up to activate that DPE ($a_{dt}$). Constraints~\ref{eq:DPEActiveInCU} and~\ref{eq:DPEActiveInDU} grant this causality for the CU and the DUs, respectively. Inequality~\ref{eq:FunctionAssign} guarantees another critical constraint: processing of all URFs ($f\in\mathcal{F}$) of each user in DPEs. Furthermore, choosing the cloud side for processing a URF depends on the delay threshold of the corresponding user ($\mu_{it}$), which is maintained by Inequality~\ref{eq:DelayConst} \footnote{In order to keep the problem complexity low, the computing costs are assumed the same for each URF. Therefore, deciding the number of URFs in one cloud side provides us a certain splitting point in the chain of the URF.}.
\par Inequalities~\ref{eq:battEnergyInCU} to~\ref{eq:renEnMaxLimitInDU} regulate renewable energy usage restrictions. The first two of them calculate the remaining energy in a battery ($b^{y}_{rt}$) according to the remaining energy from the previous time interval ($b^{y}_{r(t-1)}$), consumed green energy ($s^{y}_{rt}$), the sold energy to the grid ($p^{y}_{rt}$), and the generated renewable energy in this time interval ($S^{y}_{r}G^{y}_{rt}$). The capacity of the batteries ($B^{y}_{r}$) limits the maximum stored renewable energy, and Inequalities~\ref{eq:battLimitInCU} and \ref{eq:battLimitInDU} show this limitation for the CU and the DUs. Due to the limited battery capacity in this system, it is crucial to provide network link usage priority to the DUs, which have less remaining energy. Thus, they migrate their URFs to the CU, and the DUs that have a higher amount of renewable energy in their batteries are forced to execute URFs in their own DPEs. Therefore, we can increase the amount of renewable energy usage in the overall system. Lastly, it is clear that the consumed green energy ($s^{y}_{rt}$) can be as high as the total energy consumption in the CU and in the DUs ($\Psi^{y}_{rt}$), which are granted by Inequalities~\ref{eq:renEnMaxLimitInCU} and \ref{eq:renEnMaxLimitInDU}.
\par Equation~\ref{eq:routing} represents the flow conservation equations among the DUs and the CU. $(x,y)\in\mathcal{E}$ represents an edge that connects the nodes $x\in\mathcal{V}$ and $y\in\mathcal{V}$. If the node is a DU node $x\in\mathcal{V}^{DU}$, there should be an edge that connects this node to another node in the network. On the other hand, if the node is a CU node $x\in\mathcal{V}^{CU}$, there should be an edge that connects another node in the network to this node. For the packet-switching nodes between the DUs and CU, the summation of the number of incoming and outcoming edges should be equal to zero, which means that the path could not be disconnected at a switch node. Although this equation does not restrict a cycle in the path, a cycle does not have any beneficial effect on the objective function and the constraints. Moreover, it increases the bandwidth usage, and by adding the bandwidth constraint (Inequality~\ref{eq:bwConstraint}) to each edge, we discourage a cycle in the network.
\par The edges in this model have a limited bandwidth capacity $\Omega_{(x,y)}$ which depends on two values: the number of URFs that execute in the DPEs of the CU ($d\in D^{CU}$), and the traffic loads of each user ($\rho_{it}$). This relation is provided by Inequality~\ref{eq:bwConstraint}. However, this inequality is a quadratic constraint as a result of the multiplication of the path decision variable $l_{rt(x,y)}$, and the URF usage decision variable $m_{idft}$. Thus, we need to linearize this constraint to solve this mathematical model with a MILP solver, as shown in the next subsection.
\subsection{Problem Linearization}
\begin{linTheo}
Let $z_{rt(x,y)}$ be a continuous decision variable and can take on any value between $[0,\Omega_{(x,y)}]$. If $m_{idft}$, $l_{rt(x,y)}$, and $z_{rt(x,y)}$ decision variables satisfy Inequalities~\ref{eq:bwConstraintLinearized1} and \ref{eq:bwConstraintLinearized2}, then the quadratic constraint Inequality~\ref{eq:bwConstraint} is also satisfied by the decision variables $m_{idft}$ and $l_{rt(x,y)}$.
\begin{align}
&\sum\limits_{c\in\mathcal{C}_{r}} \sum\limits_{i\in\mathcal{I}_{c}} \sum\limits_{d\in\mathcal{D}^{CU}} \sum\limits_{f\in\mathcal{F}} \rho_{it} m_{idft} \leq M * (1-l_{rt(x,y)}) \notag \\
&+ z_{rt(x,y)}, \quad\quad\quad\quad\quad\quad\quad\quad\quad\forall r\in \mathcal{R}, \forall (x,y)\in \mathcal{E}
\label{eq:bwConstraintLinearized1}
\\
&\sum\limits_{r\in\mathcal{R}} z_{rt(x,y)} \leq \Omega_{(x,y)}, \quad\quad\quad\quad\quad\quad\quad\quad\forall (x,y)\in \mathcal{E}
\label{eq:bwConstraintLinearized2}
\end{align}
\end{linTheo}
\begin{proof}
When $l_{rt(x,y)}=0$, the first term of the right-hand side (RHS) in Inequality~\ref{eq:bwConstraintLinearized1} becomes very large due to the big $M$ value. In that case, the decision variable $z_{rt(x,y)}$ can be chosen any value between  $[0,\Omega_{(x,y)}]$ to achieve inequality. Besides, it should be noticed that a MILP solver is strongly motivated to select the lowest possible value for this variable by the sake of Inequality~\ref{eq:bwConstraintLinearized2}. 
\par For the case $l_{rt(x,y)}=1$, the first term of RHS in Inequality~\ref{eq:bwConstraintLinearized1} becomes zero. In that case, the decision value of $z_{rt(x,y)}$ should be large or equal to the total traffic of DU $r$ on edge $(x,y)$ \footnote{It is essential to remind that by Constraint~\ref{eq:routing}, there is only one path between a DU and the CU.}. Meanwhile, Inequality~\ref{eq:bwConstraintLinearized2} restricts the summation of all $z_{rt(x,y)}$ from different DUs with the bandwidth capacity ($\Omega_{x,y}$) of the corresponding edge $(x,y)$. Therefore, bandwidth capacity limitation of each edge $(x,y)$  is ensured by these two inequalities.
\end{proof}
\subsection{Complexity Analysis}
\begin{linTheo}
Minimizing the OpEx of GROVE model is an NP-Hard problem.
\end{linTheo}
\begin{proof}
(Sketch) For the static routing case, we can remove Inequalities~\ref{eq:routing} and~\ref{eq:bwConstraint}. Further, if we choose the size of solar panels and the batteries as zero,  Inequalities~\ref{eq:battEnergyInCU} to \ref{eq:renEnMaxLimitInDU} can be eliminated. Alabbasi et al. emphasize that this new reduced problem involves the bin packing problem \cite{Alabbasi2018}; thus, it is an NP-Hard problem.
\end{proof}
\begin{figure}[!bp]
\centering
\includegraphics[width=0.45\textwidth]{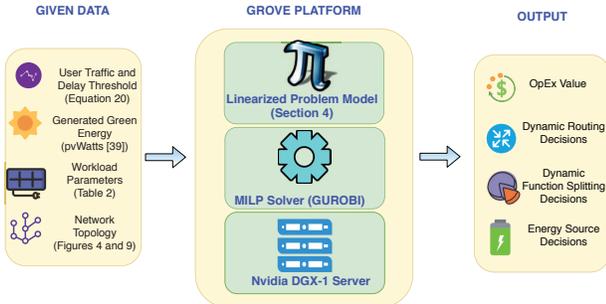}
\caption{\label{fig:expsetup} Experimental setup for evaluating the GROVE model.}
\end{figure}
\section{Computational Experiments}
\begin{figure}
\centering
\includegraphics[width=0.45\textwidth]{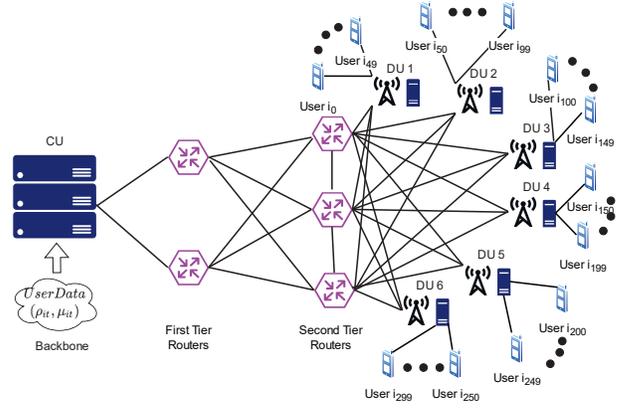}
\caption{\label{fig:top} The network topology among the central unit (CU) and the distributed units (DU) used in the experiments.}
\end{figure}  
\par This section aims to analyze the impact of jointly deciding on function splitting, routing, and RES decisions for the OpEx minimization problem of the GROVE Model.
\subsection{Evaluation Settings}
\label{sec:results_settings}
\par Figure~\ref{fig:expsetup} illustrates the experimental setup that we use to find solutions for the NP-Hard problem analyzed in the previous section. We used Gurobi \cite{GurobiOptimizationLLC2020} as a MILP solver, and computational experiments were run on an Nvidia DGX-1 Station \cite{Nvidia2020} with a Dual 20-Core Intel Xeon E5-2698 v4 2.2 GHz. The termination time was chosen as 4 hours. An example network topology used in our experiments is shown in Figure~\ref{fig:top}. It contains 12 nodes (one CU, six DUs,  and five switch nodes) and 28 edges ( ($|\mathcal{V}|,|\mathcal{E}|)=(12,28)$). First, user downlink data traffic flows from the CU to any of the routers in the first tier. Then, the data flow to any of the routers in the second tier, which are connected to the neighbor ones. In the next step, the data flows to the DU, which serves to the corresponding user. Finally, the data reaches the user either directly from the DU in an umbrella cell or from the RRH in a small cell. 
\par The number of RRHs connected to a DU equals five ($\mid\mathcal{C}_{r}\mid = 5$), and there are ten users in each RRH small cell ($\mid\mathcal{I}_{c}\mid = 10$). The delay threshold of the required service from each user ($\mu_{it}$) is randomly generated in $[0, \mid\mathcal{F}\mid]$. These users have a daily sinusoidal shape traffic load created by Equation~\ref{eq:trafficCreator} in which $\varphi$ is a random value between the $3\pi/4$ and $7\pi/4$ which defines the peak hour of the traffic profile, $\nu=3$ determines the slope of the traffic profile and $n(t)$ is a random value which produces a fluctuation in this traffic profile. In addition, we generate different peak hours for each DU to affect distinctive zones in a city such as residential, industrial, or shopping areas \cite{Pamuklu2018}. Thus, we simulate both temporal and spatial variations of a traffic load in the region of a city. Lastly, we multiply the calculated traffic load by $[0.5, 1, 1.5]$ to analyze the model for three traffic densities, which are called \textit{Low}, \textit{Medium}, and \textit{High}.
\begin{equation}
\label{eq:trafficCreator}
\lambda_{it} = \frac{1}{2^{\nu}}[1+\sin(\pi t/12 + \varphi )]^{\nu} + n(t), i\in\mathcal{I}
\end{equation}
\par Generated green energy from a solar panel ($G^{y}_{rt}$) is calculated by the pvWatts application \cite{NationalRenewableEnergyLaboratory}. We use the solar radiation data of four different cities (Stockholm, Istanbul, Cairo, Jakarta) that have a distinct distribution in a year period. Thus, we can investigate the effect of seasonal change in our model. Besides, energy prediction models may also be included easily in our system model \cite{Deruyck2018}. The rest of the parameters used in our simulations are given in Table~\ref{tab:Parameters}. The electricity price values are from Republic of Turkey Energy Market Regulatory Authorities (EPDK) variable electricity tariff regulation that has different price policies according to the time of the day \cite{EPDK2018}. The exchange rate is chosen as $1\;USD = 7\;TRY$.
\begin{table}
\centering
\caption{\label{tab:Parameters} Experiment parameters}
\begin{tabular}{|c|c|c|c|}
\hline
\textbf{Instance} & \textbf{Unit} & \textbf{CU Side} & \textbf{DU Side} \\ \hline
$E_{STA}$ & Wh & 1000 &500 \\
$E_{DPE}$ & Wh & 400 &400 \\
$S$ & kWh & 80 &20 \\
$B$ & kWh & 50 &20 \\
$\mathbb{E}_{t}$ & TRY & [0.29, 0.46, 0.70] &[0.29, 0.46, 0.70] \\
$\mathbb{P}$& - &0.5 &0.5 \\
\hline
\end{tabular}
\end{table}

\begin{figure}
\centering
\includegraphics[width=0.45\textwidth]{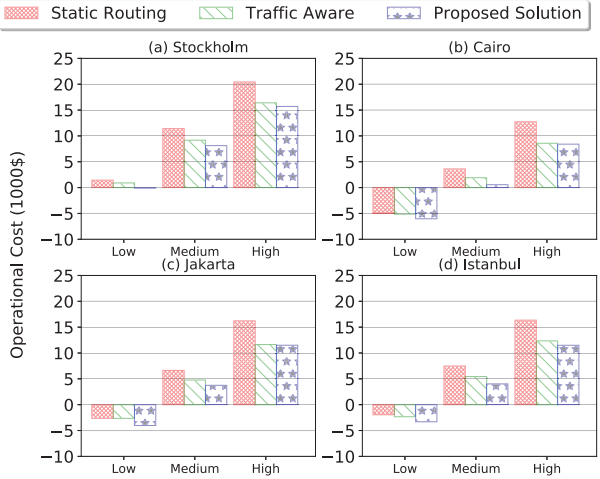}
\caption{\label{fig:bc} Operational cost of methods in different traffic loads and solar radiation distributions. (Negative values indicate that the excess renewable energy is sold back to the smart grid).}
\end{figure}
\subsection{Performance of OpEx Minimization}
\label{sec:results_opex}
\par As mentioned in the introduction section, one of the contributions of this paper is to show that we have to jointly consider the function splitting, renewable energy usage, and routing decisions. Thus, we compare the results of our proposed solution with the results of two models: a model that jointly consider the function splitting and renewable energy but does not take into account the routing decisions, called "Static Routing," and a model that take into account the function splitting and the routing decisions together but ignore the renewable energy usage decisions, called "Traffic-Aware."
\par Figure~\ref{fig:bc} shows the one year period OpEx results for different traffic loads and solar radiation distributions. The results confirm that our proposed solution has lower OpEx for any traffic load and solar radiation distribution. Thus, we can use this model for any city or urban region around the world. Meanwhile, Figure~\ref{fig:bc} also shows that with higher traffic loads, the OpEx increases due to the boosting of the number of active DPEs\footnote{The reason for some negative OpEx in the low traffic load is that the system's profit gain of selling renewable energy is higher than the grid energy bills in this level of traffic load for this particular setup corresponding to a specific CapEx.}. Also, the cities that have higher solar radiation rates have lower OpEx owing to the increase in the renewable energy availability.
\begin{figure}
\centering
\includegraphics[width=0.45\textwidth]{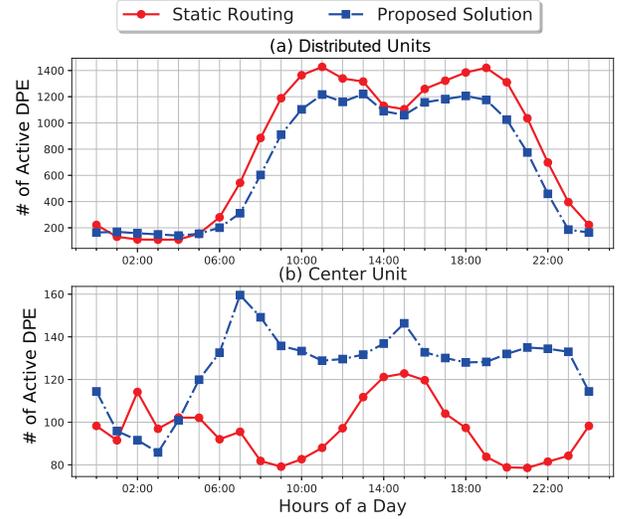}
\caption{\label{fig:du} Number of active DPEs in a day period in CU and DUs (average of 12 different scenarios).}
\end{figure}
\begin{figure}
\centering
\includegraphics[width=0.45\textwidth]{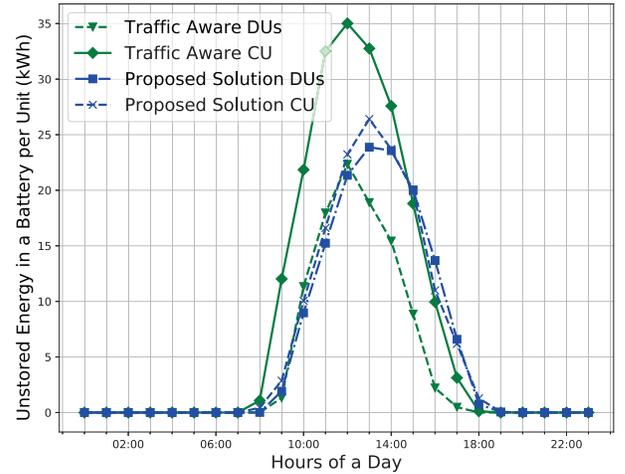}
\caption{\label{fig:unstored} Unstored energy in DUs and CU in a day period.}
\end{figure}
\par If we analyze why our proposed solution provides better results than the static routing, we have to focus on the number of active DPEs in the DUs and CU. As we can see in Figure~\ref{fig:du}, our proposed solution decreases the number of active DPEs in the DUs by choosing the routing decisions efficiently. Thus we can reduce the overall energy consumption by centralizing the functions in the CU. In the meantime, our proposed solution also beats the "Traffic-Aware" algorithm by activating the DPEs of DUs which have renewable energy in their batteries. Thus we can use this valuable renewable energy efficiently and prevent the unstored energy as shown in Figure~\ref{fig:unstored}. Moreover, our solution reserves the renewable energy of the DUs for more profitable hours (Figure~\ref{fig:rem}). Therefore, we can reduce the OpEx further by considering the renewable energy in the batteries of the DUs and the CU.
\begin{figure}
\centering
\includegraphics[width=0.45\textwidth]{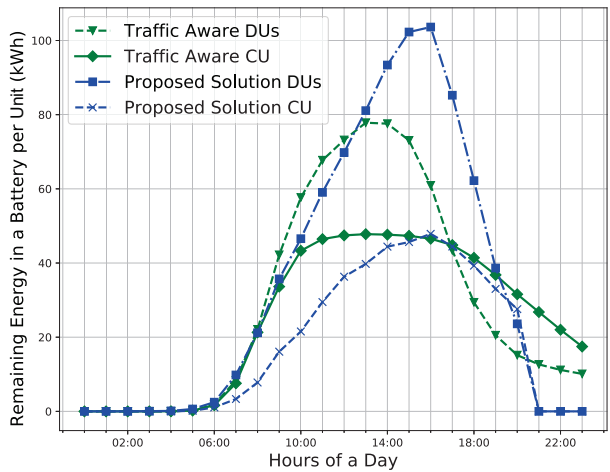}
\caption{\label{fig:rem} Remaining energy in the CU and DUs in a day period.}
\end{figure}
\begin{figure*}
\centering
\subfloat[\label{fig:tpc1} 6 DU topology.]{
\includegraphics[width=0.33\textwidth]{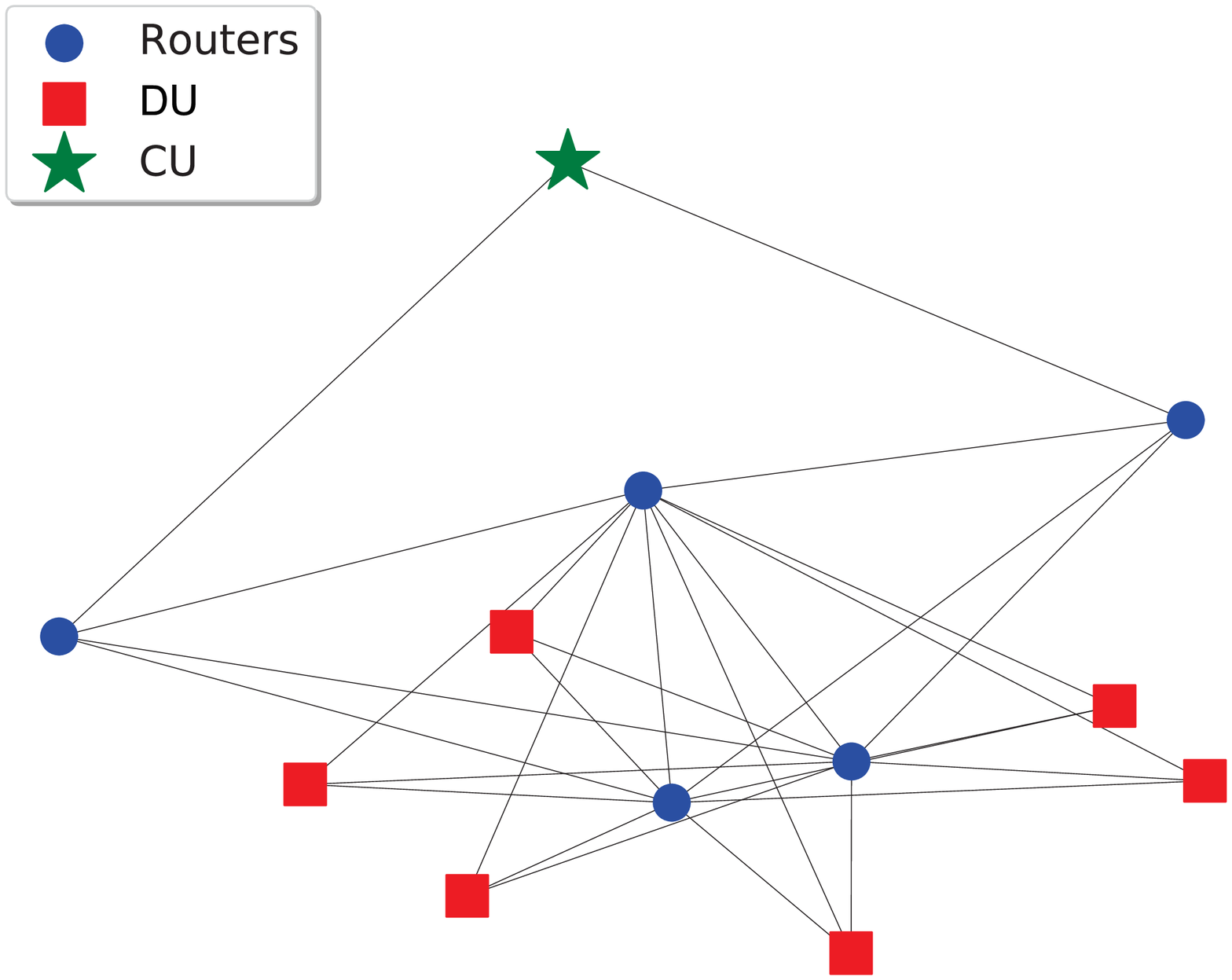}}
\subfloat[\label{fig:tpc2} 12 DU topology.]{
\includegraphics[width=0.33\textwidth]{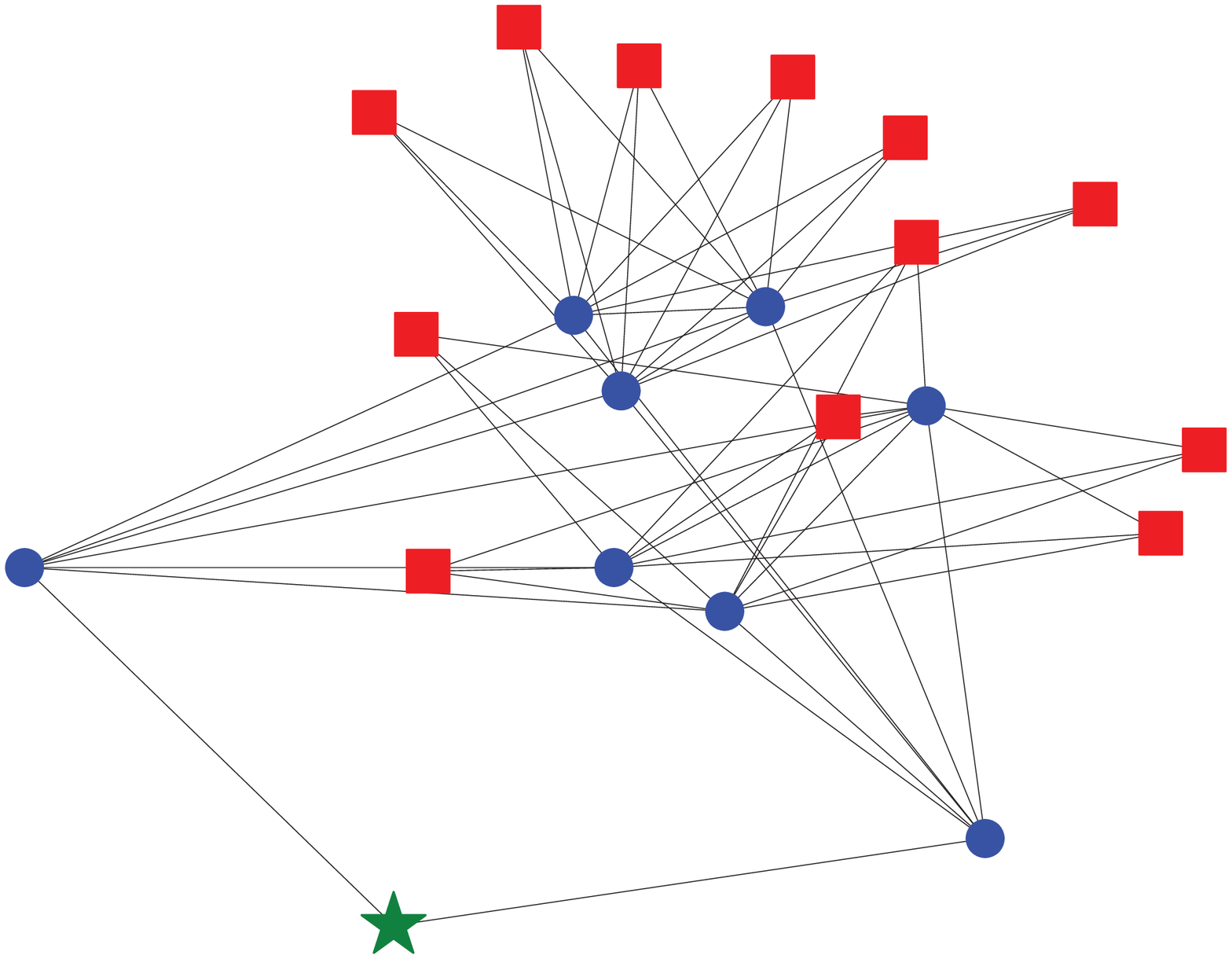}}
\subfloat[\label{fig:tpc4}24 DU topology.]{
\includegraphics[width=0.33\textwidth]{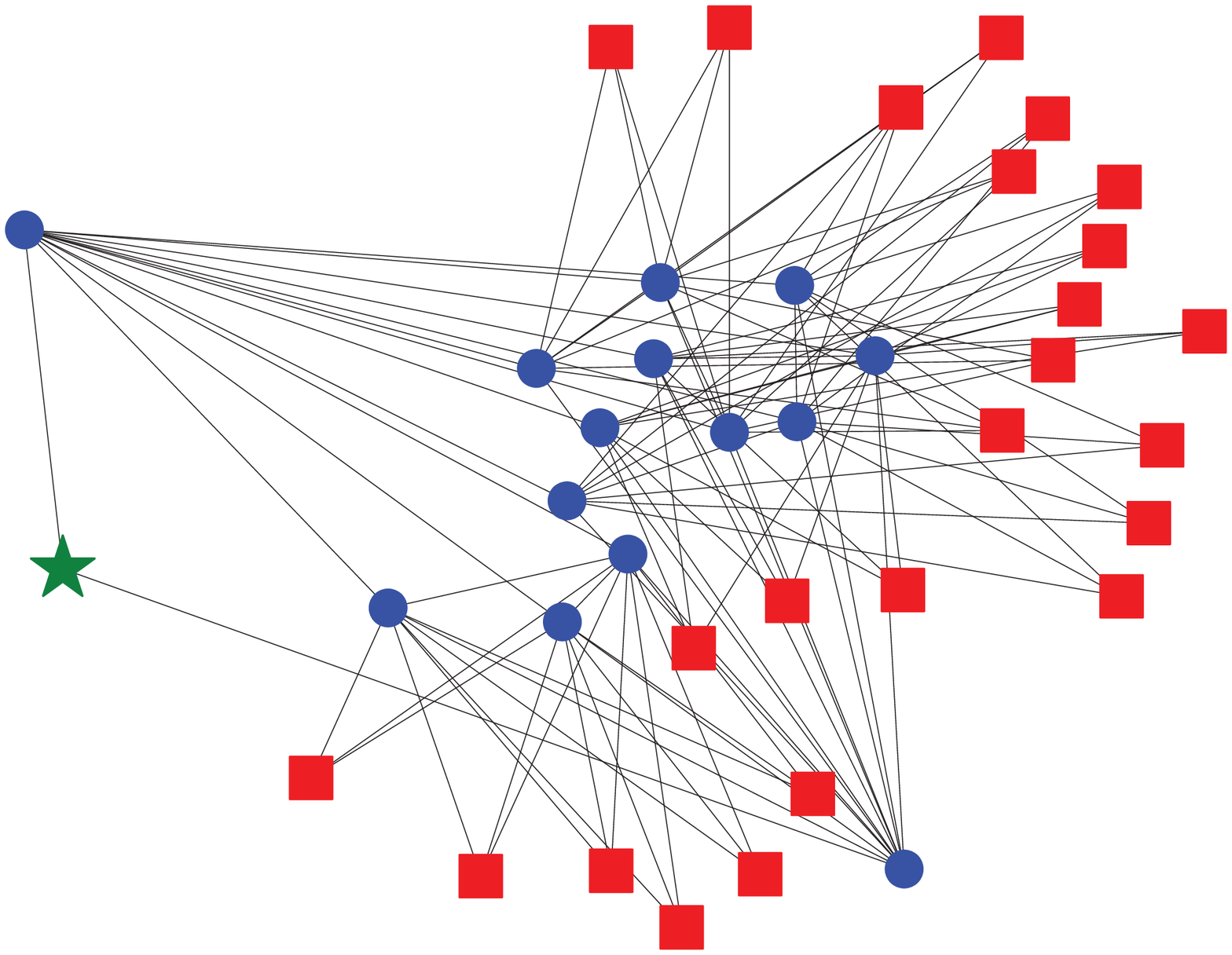}}
\caption{Network scalability study: Network topologies with different number of nodes and edges.}
\label{fig:networkScale}
\end{figure*}
\begin{figure}
\centering
\includegraphics[width=0.45\textwidth]{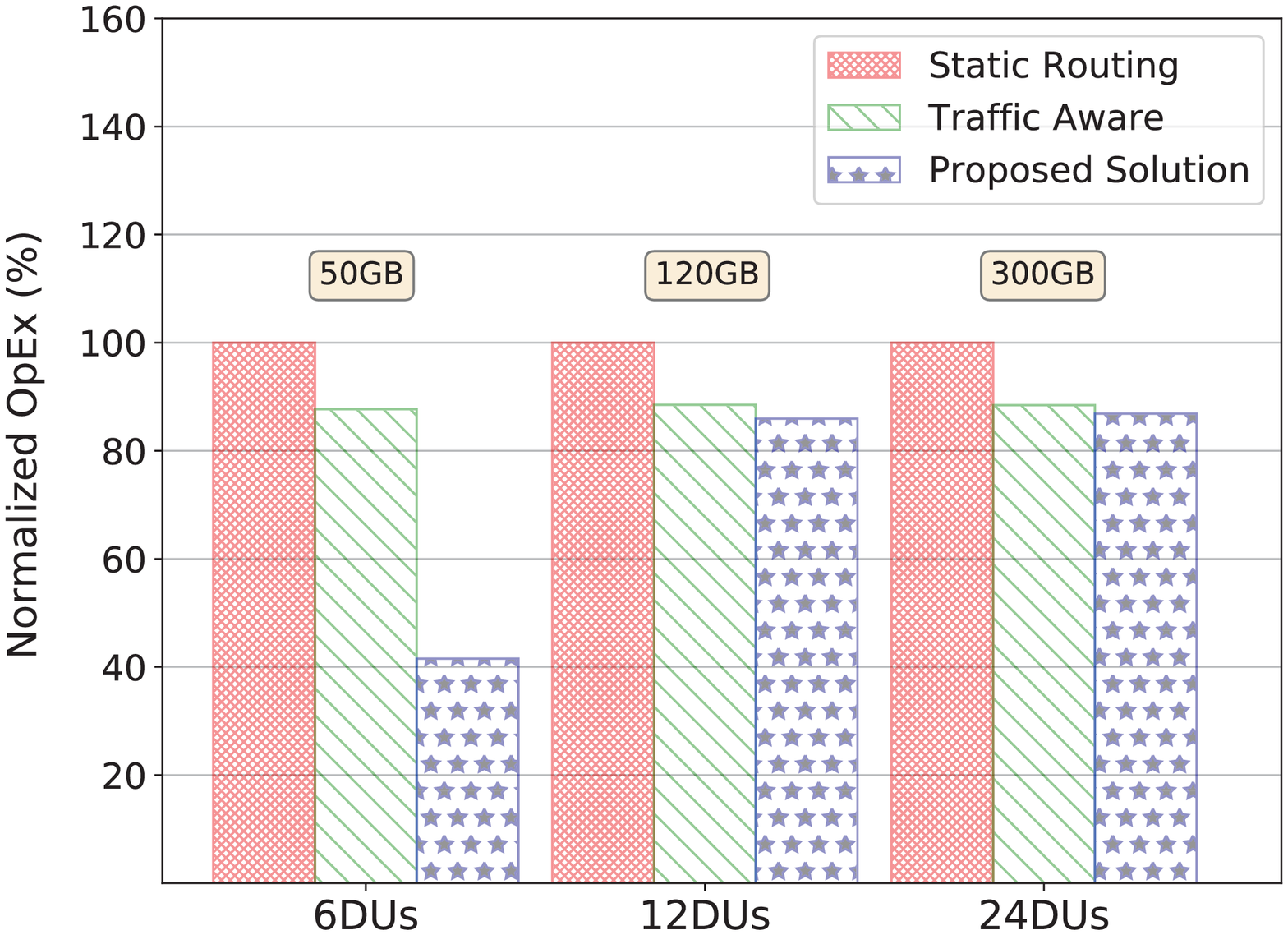}
\caption{\label{fig:nscomp} OpEx and RAM consumption of the MILP solver for different network topologies (Istanbul, medium traffic rate)}
\end{figure}
\subsection{Network Scalability Analysis}
\par The example network configuration we studied was described in Section~\ref{sec:results_settings}, and its results were presented in the previous subsection. Now, we will analyze the outcomes with larger topologies. Figure~\ref{fig:networkScale} shows  three network topologies that have different sizes. Figure~\ref{fig:tpc1} is our initial topology, which has six DUs and five switch nodes connecting them to the CU. In Figure~\ref{fig:tpc2}, we doubled the number of DUs (12 DUs), and then we used eight switches to connect them to the center. Lastly, Figure~\ref{fig:tpc4} shows an architecture that has 24 DUs and  14 switches.
\par  Due to the non-polynomial difficulty of the optimization problem, as the network size gets larger, the required memory (As indicated in Figure~\ref{fig:nscomp}) and computational power are increasing very rapidly. Thus, the size of the third topology is a boundary for the GROVE Model by using the DGX-1 Station for the given run time limit. Moreover, the increase in the network size reduces the performance of our proposed solution. The main reason for this outcome is related to the size of the solution space of the compared solutions. Our proposed solution has a more extensive solution space; thus, a MILP solver needs more processing time to approximate the lower bound. Besides, an MNO gets better results with our proposed solution for larger network sizes, according to the results shown in Figure~\ref{fig:nscomp}.
\section{Conclusion}
\par We propose a novel network model named GROVE, which connects the CU and DUs on a mesh network, as a cost-efficient solution for using RESs in a C-RAN architecture. Then, we formulate an OpEx minimization problem that jointly takes into account the function splitting, routing, and RES usage decisions. Coexistence of these decisions at the same time interval results in a quadratic programming problem. We linearize the quadratic constraints to solve this problem with a MILP Solver. The results show that our model improves the disjoint approaches' performance, and it is more feasible for different solar radiation distributions and various traffic densities. The proposed solution outperforms the static routing by using the multiplexing gain on the mesh network. Also, it performs better than the traffic-aware method by optimizing the usage of profitable renewable energy in the batteries of the CU and DUs.
\par The network scalability analysis shows that a MILP solver can maintain a network model with 40 nodes with reasonable RAM consumption. Besides, our proposed solution performs better results even with larger networks, and we may increase the termination time to reduce the gap between the lower bound of the MILP solver to improve this performance. For even larger problems, we plan to develop a heuristic using lower CPU processing and memory resources. We will investigate the economic benefits of constructing an RoE network as an alternative to a standard network with dedicated links. Hence, we are planning to encourage MNOs to implement the proposed RESs system to their next-generation networks to improve their cost and energy efficiency.
\bibliography{grove}
\end{document}